\documentclass[a4paper]{amsart}
\usepackage{amsfonts}
\usepackage{amsmath}
\usepackage{amssymb}
\usepackage{amsthm}
\usepackage{xcolor}
\usepackage{microtype}

\newcommand{\sed}[1]{}

\newcommand{\ga}{\overline \gamma}
\newcommand{\emp}{\varepsilon}

\newcommand{\sff}{\Sigma^*/_{\ff}}

\newcommand{\N}{\mathbb N}

\newcommand{\ff}{\sim}
\newcommand{\ph}{\varphi}

\newcommand{\fin}{\mathrm{Fin}}
\newcommand{\fins}{\mathrm{Fin}(\Sigma^*)}

\newcommand{\gn}[1]{\left\langle #1 \right\rangle}
\newcommand{\abs}[1]{\left| #1 \right|}

\newtheorem{theorem}{Theorem}
\newtheorem{lemma}{Lemma}

\theoremstyle{definition}
\newtheorem{definition}{Definition}

\theoremstyle{remark}
\newtheorem{example}[theorem]{Example}
\newtheorem{remark}[theorem]{Remark}

\begin{document}

\title{Pseudo-solutions of word equations}

\author{\v St\v ep\' an Holub}
\address{Charles University in Prague Department of Algebra, Czech Republic}
\email{holub@karlin.mff.cuni.cz}


\begin{abstract}
We present a framework which allows a uniform approach to the recently introduced concept of pseudo-repetitions on words in the morphic case. This framework is at the same time more general and simpler. We introduce the concept of a pseudo-solution and a pseudo-rank of an equation. In particular, this allows to prove that if a classical equation forces periodicity then it also forces pseudo-periodicity. Consequently, there is no need to investigate generalizations of important equations one by one. 
\end{abstract}

\maketitle

\date{}

\section{Introduction}
It is one of the most basic properties of words (or strings) that if two words $x$ and $y$ satisfy $xy = yx$, then they are both repetitions of a single word. Using a different terminology to express the same fact, the equation $xy = yx$ is `periodicity-forcing'.
Periodicity forcing equations are of a special interest since they delimit nontrivial relations that words can satisfy. Also, it is typically quite difficult to show that an equation is periodicity forcing. A famous example is the equation of Lyndon and Sch\"utzenberger, namely $x^ay^b = z^c$, which is periodicity-forcing when $a,b,c \geq 2$ \cite{ls}. This result has a lot of generalizations, see e.g. \cite{dirkLS2005}.

A generalized concept of pseudo-repetition was introduced in \cite{Elena2010}, involving antimorphic involution motivated by DNA complementarity. This brought about a new group of questions. For example, a complete solution of the Lyndon and Sch\"utzenberger equation in this generalized setting was given in \cite{PseudoLSComplete}.
In \cite{ManeaFW2012}, questions were further generalized by considering arbitrary morphic and antimorphic permutations. In the case of morphic permutations, it turned out that several classical results hold even in the generalized setting, while much less can be said about the antimorphic case. In \cite{DayWords2017}, the authors announced, as the main result, a solution of the Lyndon and Sch\"utzenberger equation for morphic permutations. The proof is only sketched and is very involved. 

 The aim of this paper is to settle a large group of questions concerning pseudo-repetitions in the case of morphic permutations once and for all by showing that \emph{all} equations that are periodicity forcing are also pseudo-periodicity forcing. Example \ref{ex2017} shows how this can be helpful.   In fact, our result is more general in two respects. First, it holds in a more general setting of which the morphic permutation is just one special case. Second, the analogous result holds for any rank of the equation, where `periodicity-forcing' means `of rank one'.   
The strong point of our approach is that this level of generality actually makes the argument simpler.

\section{Anticongruences}
Let $\Sigma$ be an alphabet. The set of all words over $\Sigma$ is denoted as $\Sigma^*$. Endowed with the operation of concatenation, denoted by $\cdot$\,, $\Sigma^*$ is a monoid where the empty word $\emp$ is the unit. 

  The most general setup in which our approach works is captured by an equivalence $\ff$ on $\Sigma^*$ satisfying the following conditions:
\begin{enumerate}
	\item if  $u \ff v$, then $|u| = |v|$;
	\item if $u\cdot v \ff u'\cdot v'$ and $|u|=|u'|$, then $u\ff u'$ and $v\ff v'$.  \label{anti2}
\end{enumerate}
By the first condition, the equivalence $\ff$ can be understood as a countable collection of equivalences, each defined on words of given length. The second condition is in a sense opposite to the familiar property of a congruence, and that is why we call an equivalence satisfying the two conditions above an \emph{anticongruence}. 

Let $[u]$ denote the equivalence class of a word $u$. It is the definition of a monoid congruence that $[u]\odot [v] \subseteq [u\cdot v]$, where $[u]\odot[v] =\{x\cdot y \mid x \in [u],\ y \in [v]\}$. In contrast, for an anticongruence, the condition \eqref{anti2} implies $[u\cdot v] \subseteq [u]\odot [v]$, and the inclusion may be strict. In fact, we have
\begin{align*}
[w] \subseteq \bigcap_{u_i\cdot v_i = w} [u_i]\odot [v_i]\,,
\end{align*}
where $0 < i < |w|$ and $\abs{u_i}=i$.
Note also that 
\[ [u]\odot [v] =\bigcup_{\begin{minipage}[c]{3em} $u'\sim u$\\ $v'\sim v$	
		\end{minipage}}[u'\cdot v']\,, \]
which means that the set $[u]\odot [v]$ typically consists of several equivalence classes.
\begin{example}
Let $\Sigma= \{a,b,c,d\}$, and $\ff$ be such that $a\ff c$ and $b\ff d$, and no other pairs of letters are equivalent. Then, for example, $ac \not\ff ba$. On the other hand, $ac \ff ca$ may, but \textbf{need not} be true. Let $aa\ff cc$ and  $ac \not\ff ca$. Then $aacc\not\ff ccaa$ because $acc \not\ff caa$ which in turn follows from $ac\cdot c \not\ff ca \cdot a$.
\end{example}

\begin{example}
	The trivial instance of an anticongruence is equality. A more interesting example, which motivated our research, is given by \emph{morphic permutations} introduced and studied in \cite{DayWords2017}. A morphic permutation is a permutation on $\Sigma$ extended to a (length preserving) morphism on $\Sigma^*$. Given such a morphic permutation $f$ we can define an anticongruence by $u \ff v$ iff $u = f^i(v)$ for some $i\geq 0$. 	
	The equivalence classes of a letter under a morphic permutation equivalence are cycles of the permutation. Note also that if $u' = f^i(u)$ and
	$v' = f^j(v)$ then $u\cdot v \ff u'\cdot v'$ if and only if $i \equiv j \mod n$ for each  size $n$ of the cycle for some letter in $u$ or $v$.
\end{example}

\section{Pseudo-free basis}
Our results are based on a generalization of some standard notions and facts from free monoids to a setting involving an anticongruence.  
For those generalizations, we add the prefix ``pseudo-'' used in the term ``pseudo-repetition'' introduced in \cite{Elena2010}. 
 
From now on, let an anticongruence $\ff$ on $\Sigma^*$ be given. We shall often refer to the set $\sff$ of equivalence classes of $\ff$. The set $\sff$ is a subset of $\fins$, which is a monoid of finite languages over $\Sigma$ with the operation $\odot$ defined as above by $K \odot L = \{u\cdot v \mid u\in K, v \in L\}$. 
Note, however, that $\sff$ itself is in general not a submonoid of $\fins$, since $[u]\odot [v]$ may consist of several distinct classes of $\ff$.

Let $C$ be a set of equivalence classes, that is, $C\subseteq \sff$. We shall consider two different structures given by $C$. The first one is the submonoid $\gn{C}$ of $\fins$ generated by $C$.
Second, we can see $C$ as an alphabet, and consider the free monoid $C^*$ over $C$. Elements of $C^*$ are best understood as lists $(c_1,c_2,\dots,c_\ell)$ endowed with the operation of list concatenation.
As usual, we shall identify  a `word' $(c)\in C^*$ of length one  with the `letter' $c\in C$.
Note again that although elements of $C$ are equivalence classes of $\ff$, a set $c_1\odot c_2 \cdots \odot c_m\in \gn C$ typically decomposes into many distinct classes.

We shall be interested in \emph{$\ff$-closed monoids} of words, that is, monoids $M\subseteq \Sigma^*$ such that $[u]\subseteq M$ for each $u\in M$. For a free monoid $M\subseteq \Sigma^*$ which is $\ff$-closed, we have the following claims. 
\begin{lemma}\label{uniquec}
	Let $B$ be the (free) basis of a free $\ff$-closed monoid $M$, and let $C = \{[b] \mid b \in B\}$.  
	
	For each $w\in M$, there is a unique word $(c_1,c_2,\dots,c_m)\in C^*$ such that $w\in c_1 \odot c_2 \odot \cdots \odot c_m$. Moreover, $[w]\subseteq c_1 \odot c_2 \odot \cdots \odot c_m$. 
\end{lemma}	
\begin{proof}
	We first show that $[b]\subseteq B$ for each $b\in B$. Let $u\in [b]$, and let $u = b_1b_2\cdots b_k$ be the (unique) factorization of $u$ into elements of $B$. Let $b = b'_1b_2'\cdots b_k'$ where $|b_i|=|b_i'|$, $i=1,2,\dots,k$. Since $u\ff b$, we have $b_i \ff b_i'$ and therefore $b_i'\in M$, $i=1,2,\dots,k$. This implies $k=1$, and $u\in B$.

	Let now $w = b_1b_2\cdots b_m$ be the (unique) decomposition of $w$ into elements of $B$. Then $w \in c_1 \odot c_2 \odot \cdots \odot c_m$ where $c_i = [b_i]$.
Assume that \[w\in c_1 \odot c_2 \odot \cdots \odot c_m \cap c'_1 \odot c'_2 \odot \cdots \odot c'_{m'},\]	
where each $c_j'$ is an element of $C$. Then $w = b_1b_2\cdots b_m = b_1'b_2'\cdots b'_{m'}$ for some $b_i\in c_i$, $b_j'\in c_j'$. This implies $m = m'$ and $c_i  = c_i'$.

Let finally $w'\ff w = b_1b_2\cdots b_m$ and let $w'= w_1'w_2'\cdots w_m'$ where $|w_i| = |b_i|$. Then $w_i' \ff b_i$, hence $w'\in c_1 \odot c_2 \odot \cdots \odot c_m$.
\end{proof}
For $M$, $B$ and $C$ as in the previous lemma,  we say that $C$ is the \emph{pseudo-free} basis of $M$. By the lemma, we can define a mapping $\gamma: M \to C^*$ by $\gamma: w \mapsto (c_1,c_2,\dots, c_m)$ where $w\in c_1 \odot c_2 \odot \cdots \odot c_m$. It is easy to see that $\gamma$ is a monoid morphism, that is, $\gamma(w_1w_2) = \gamma(w_1)\gamma(w_2)$.
We shall call $\gamma$ the \emph{pseudo-free factorization} of $M$.

The intersection of two free monoids is well known to be free (see \cite[Chapter 1, Section 1.2]{lothaire}). Also, the intersection of two $\ff$-closed monoids is again $\ff$-closed. Therefore, the smallest (w.r.t. inclusion)  $\ff$-closed free monoid containing a given set $W$ of words exists. Such a monoid  $M$ is called the \emph{pseudo-free hull} of $W$,  and the cardinality of its pseudo-free basis is the \emph{pseudo-rank} of $W$.

We remark that standard notions of the \emph{free hull} and of the \emph{rank} of a set $W$ are obtained as above if we consider the trivial case of identity for $\ff$.

\section{Equations and pseudo-solutions}

We now turn our attention to equations. 
A word equation (without constants) is a pair $(r,s)\in \Theta^*\times \Theta^*$ where $\Theta$ is an alphabet of unknowns. (Often, $r=s$ is written instead of $(r,s)$.) A solution  (in $\Sigma^*$) of $(r,s)$ is a monoid morphism $\ph:\Theta^*\to \Sigma^*$ such that $\ph(r)=\ph(s)$.
In other words, $\ph$ is a solution  of $(r,s)$ if the words $r$ and $s$ become equal after substituting each $x\in \Theta$ with $\ph(x)$. Informally, an equation can be seen as a relation and a solution gives words satisfying that relation. 

The \emph{rank of a solution} $\ph$ is the rank (that is, the size of the free basis) of the set $\ph(\Theta)=\{\ph(x) \mid x \in \Theta \}$. Finally,
the \emph{rank of the equation} $e$ is the maximum rank of a solution $\ph$ of $e$. 

\begin{remark}
The rank of the solution defined above is usually called the `free rank'. It is possible to define other ranks of a set of words. For example, the `combinatorial rank' is easier to define but it is less suitable for our purposes. An important fact is that the rank of an equation is independent of what kind of rank is considered. See \cite[Chapter 6, Section 5.2]{handbook} for more details. 
\end{remark}

\begin{example}
	Consider the equation $e=(xy,zx)$ over $\Theta = \{x,y,z\}$. Then $\ph: x \mapsto a$, $y \mapsto bca$, $z \mapsto abc$ is a solution of $e$ in $\{a,b,c\}^*$. 
	The free basis of $\{a,bca,abc\}$ is $\{a,bc\}$ and the rank of $\ph$ is therefore two.
    It is known that any set of three words satisfying a nontrivial relation have the rank at most two (this is a special case of the so called `defect effect'). Consequently, the rank of $e$ is two.
\end{example}

Equations of rank one are called \emph{periodicity forcing}. 

\begin{example}\label{commute}
	It is not difficult to see that two words $u,v\in \Sigma^*$ satisfy $uv = vu$ if and only if $u=t^i$ and $v=t^j$ for some $t\in \Sigma^*$ and $i,j\in \N$. Therefore, $(xy,yx)$ is a periodicity forcing equation. 
\end{example}

We now explain what a \emph{pseudo-solution} is. Informally, in a pseudo-solution, we are allowed, in order to achieve the equality, to substitute different occurrences of the same unknown with different, but equivalent words.  
More precisely, if $(x_1x_2\cdots x_\ell,\,x_{\ell+1}x_{\ell+2}\cdots x_{\ell+m})$ is an equation, then its pseudo-solution is given by words  $u_1,u_2,\dots, u_{\ell+m} \in\Sigma^*$ such that 
$u_i \ff u_j$ if $x_i = x_j$, and $u_1u_2\cdots u_\ell = u_{\ell+1}u_{\ell+2}\cdots u_{\ell+m}$.

Mostly as an exercise, we shall show that it is equivalently possible to use a more relaxed requirement, namely that after the substitution, the left and right side of the equation are equivalent. 

\begin{lemma}\label{eqeq}
	Let $(x_1x_2\cdots x_\ell,\,x_{\ell+1}x_{\ell+2}\cdots x_{\ell+m})$ be an equation. Let $u_1,u_2,\dots, u_{\ell+m}$ be words over $\Sigma$ such that 
	$u_i \ff u_j$ if $x_i = x_j$, and $$u_1u_2\cdots u_\ell \ff u_{\ell+1}u_{\ell+2}\cdots u_{\ell+m}.$$ Then there are words $u_1',u_2',\dots, u_{\ell+m}'$ such that 
	$u_i' \ff u_j'$ if $x_i = x_j$, and $$u_1'u_2'\cdots u_\ell' = u_{\ell+1}'u_{\ell+2}'\cdots u_{\ell+m}'.$$
\end{lemma}
\begin{proof}
	Set $u_i' = u_i$ for $i = 1,2,\dots, \ell$, and let $w = u_1u_2\cdots u_\ell$. We define $u_j'$, $j=\ell+1,\ell+2,\dots,\ell+m$ as factors of $w$ such that $|u_j'|=|u_j|$, and  $w=u_{\ell+1}'u_{\ell+2}'\cdots u_{\ell+m}'$. The definition of an anticongruence yields that $u_j'\ff u_j$, which concludes the proof.
\end{proof}

Since the above definition of a pseudo-solution is a little bit too verbose, we adopt the following formal definition. It should be obvious that it is equivalent to the above description.

\begin{definition}
We say that a monoid morphism $\ph: \Theta^* \to \fins$ is a \emph{pseudo-solution} of $(r,s)$ if
\begin{itemize}
	\item  $\ph(x) \in \sff$ for each $x\in \Theta$; and
	\item  $\ph(r) \cap \ph(s)\neq \emptyset$. 
\end{itemize}
\end{definition}
By Lemma \ref{eqeq}, the second condition could be replaced with 
\begin{itemize}
	\item there exist $w_1\in \ph(r)$ and $w_2\in \ph(s)$ such that $w_1 \ff w_2$. 
\end{itemize}

Finally, we define the \emph{pseudo-rank of the equation} $e$ as the maximum pseudo-rank of $\bigcup \ph(\Theta)$ over all possible pseudo-solutions $\ph$ of $e$. An equation of pseudo-rank one is called \emph{pseudo-periodicity forcing}.

\begin{example}
	Consider the equation $e=(xy,yx)$. Let $\Sigma=\{a,b\}$, $a\ff b$, and for words of length at least two, $u \ff v$ is equivalent to $u=v$. Then $\ph_1$ with $\ph_1(x)=[a]$ and $\ph_1(y)=[b]$ is a pseudo-solution of $e$. In fact, $\ph_1(x)=\ph_1(y)=\{a,b\}$ and $\ph_1(xy)=\ph_1(yx)=\{aa,ab,ba,bb\}$. 
	
	On the other hand, the morphism $\ph_2$ defined by $\ph_2(x)=[ab]$ and $\ph_2(y)=[a]$ is not a pseudo-solution of $e$, since $\ph_2(xy)=\{aba,abb\}$ and 
	$\ph_2(yx)=\{aab,bab\}$.
	
	A simple application of our main result below yields that the pseudo-rank of $(xy,yx)$ is one, since it is periodicity forcing. 
\end{example}

\begin{example}\label{exph}
	Consider the equation $e=(xyz,zyx)$. Let $\Sigma=\{a,b,c\}$, and let $\ff$ be the smallest anticongruence satisfying
	\begin{itemize}
		\item $a\ff c$
		\item $ab\ff cb$, $bc \ff ba$
		\item $abc \ff cba$ 
	\end{itemize}
	Then $\ph(x) = [abc]$, $\ph(y) = [b]$, $\ph(z)=[a]$ is a pseudo-solution of $e$. We have
	\begin{align*}
	\ph(xyz) & = \{abc\cdot b \cdot a,cba\cdot b \cdot a,abc \cdot b \cdot c,cba\cdot b\cdot c\}, \\
	\ph(zyx) & = \{a\cdot b\cdot abc,c\cdot b\cdot abc,a\cdot b\cdot cba,c\cdot b\cdot cba\},
	\end{align*}
	hence $abcba \in \ph(xyz) \cap \ph(zyx)$.
	The pseudo-free basis of $\bigcup \ph(\Theta)$ is $\{[a],[b]\}$, and the pseudo-rank of $e$ is two.
\end{example}

\section{Main result}

Periodicity forcing equations have only solutions that can justly be called trivial, namely solutions in which all images are powers of a single word. In Example \ref{commute}, we mentioned that commutation is periodicity forcing. In fact, it is well known that any equation $(r,s)$ in two unknowns is periodicity forcing if it is nontrivial, that is, if $r\neq s$. Concerning three unknowns, we mentioned in the Introduction the famous result by Lyndon and Sch\"uzenberger which claims that $(x^ay^b,z^c)$ is periodicity forcing as soon as $a, b, c\geq 2$. Less well known is the important generalization which says that $(w,z^c)$, where $w\in \{x,y\}^*$ is a primitive word and $c\geq 2$, is periodicity forcing unless $w\in x^*yx^*$ or $w=y^*xy^*$ (see \cite{spehner} and \cite{lerest}).
For some results concerning four unknowns, see for example \cite{CMUC}.

For each of those results a related question can be asked:
Given an equation $e$ that is periodicity forcing,  is it true that all pseudo-solutions are pseudo-periodic? Of course, this may depend on the anticongruence in question. Some answers for particular equations and morphic permutations were given in \cite{ManeaFW2012,DayWords2017}.
The following theorem gives a general answer which makes it unnecessary to produce new results of this kind.

\begin{theorem}\label{main}
	The pseudo-rank of an equation $e$ is at most its rank.
\end{theorem} 
\begin{proof}
	Let $\ph$ be a pseudo-solution of an equation $e = (r,s)$. 
	Let $M$ be the pseudo-free hull of $\bigcup\ph(\Theta)$, let $C$ be the pseudo-free basis of $M$, and let $\gamma$ be the pseudo-free factorization of $M$. 	Note that the cardinality of $C$ is the pseudo-rank of $\ph$. 
			
	We first adopt the following notation. Let $S\subseteq M$ be such that $\gamma$ has the same value for all elements of $S$. Then we shall write $\ga(S)$ for $\gamma(s)$, $s\in S$. In particular, $\ga([w])$ is well defined in this way for any $w\in M$ by Lemma \ref{uniquec}. Lemma \ref{uniquec} also implies that $\ga(S\odot T) = \ga(S)\ga(T)$ if both $\ga(S)$ and $\ga(T)$ are defined. Altogether, this convention defines a morphism $\ga: \gn C \to C^*$.
	
	We can now define a morphism $\alpha:\Theta^* \to C^*$ by $\alpha(x) = \ga(\ph(x))$. Using the word $w\in \ph(r)\cap \ph(s)$, we also see that $\alpha(r)= \alpha(s) = \gamma(w)$. Therefore, $\alpha$ is an (ordinary) solution of $e$. 
	
	It remains to show that the rank of $\alpha$ is equal to the pseudo-rank of $\ph$, that is, to the cardinality of $C$. In other words, we have to show that $C$ is the free basis of $\{\alpha(x) \mid x\in \Theta\}$.
	
	Let $Y\subseteq C^*$ be the free basis of $\{\alpha(x) \mid x\in \Theta\}$ and let $Z=\{z\in M \mid \gamma(z)\in Y\}$.  We claim that $\gn{Z}$  is freely generated by $Z$. Indeed, let $w=z_1z_2\cdots z_m = z_1'z_2'\cdots z_k'$ where $z_i,z_j'\in Z$ and $z_1\neq z_1'$. Then $|z_1|\neq |z_1'|$, which implies that $\gamma(z_1)\neq \gamma(z_1')$, a contradiction with $\gn Y$ being free since $\gamma(z_1)\gamma(z_2)\cdots \gamma(z_m) = \gamma(z_1')\gamma(z_2')\cdots \gamma(z_k')$.  
	By Lemma \ref{uniquec}, $\gn Z$ is a $\ff$-closed monoid. Also, $\ph(x)\subseteq \gn Z$ for each $x\in \Theta$.
		The minimality of $M$ now implies $\gn Z=M$, hence $Y=C$.
\end{proof}

\begin{example}
	In Example \ref{exph}, we have $C = \{[a],[b]\}$, $\alpha(x) = ([a],[b],[a])$, $\alpha(y) = ([b])$ and $\alpha(z) = ([a])$.  
\end{example}
	
\begin{example}\label{ex2017}
	In \cite{DayWords2017}, the following claim is proved (see \cite{DayWords2017}, Eq.1 and Theorem 21):\medskip
\begin{quote}
	Let $u,v,w\in \Sigma^+$ satisfy 
\[ f^{a_1}(u)\cdots f^{a_r}(u) f^{c_1}(v)\cdots f^{c_s}(v) = f^{b_1}(w)\cdots f^{b_t}(w) \]
for $r,s,t \geq 2$, and a morphic permutation $f$. Then $u$, $v$, and $w$ are $[f]$-repetitions.
\end{quote}	

Our contribution with respect to the above claim is threefold. 
First, we can formulate the claim more elegantly as 
\begin{quote}
	The equation $(x^ay^b,z^c)$ with $a,b,c \geq 2$ is pseudo-periodicity forcing.
\end{quote}	

Second, the claim is more general, since it holds for all anticongruences, not only for morphic permutations. 

Third, and most importantly, the claim is a direct consequence of Theorem \ref{main}, whose proof should be compared to the rough sketch of the proof in \cite{DayWords2017}. Of course, this does not mean that we have found a short self-contained proof of the above claim. The complexity of the whole proof is given by the complexity of the classical proof. However, we have shown that there is no need to laboriously reconstruct the classical proof in a more tedious setting. 
\end{example}	

\begin{example}
	Our approach does not work for anti-morphic permutations for an obvious reason: an analog of Theorem \ref{main} does not hold in that case. For example, let $u$ be equivalent to its reversal $\overline u$. Then the equation $(x^2y^2,z^4)$ has a `pseudo-solution' $x \mapsto aabaaab$, $y\mapsto a$, $z\mapsto aaba$, since
	\[ aabaaab \cdot aabaaab \cdot a \cdot a = aaba \cdot aaba \cdot \overline{aaba} \cdot \overline{aaba}\,.  \]
	In the same time, the rank of this solution is not one, since the words $aabaaab$, $a$ and $aaba$ are not contained in $\gn{\{t,\overline t\}}$ for any $t$. (For more similar examples cf. \cite[Example 1]{Czeizler:2011:ELR:1950991.1951213}.)
\end{example}

\section{Final comments}
We give some less formal comments that can provide further insight into why Theorem \ref{main} holds. 

The most simple nondeterministic algorithm yielding a solution of an equation is as follows: Given an equation $\big((x \cdots),(y \cdots)\big)$, guess which of the images of $x$ and $y$ will be longer, and perform what is called an elementary transformation of the equation. Say that the image of $x$ will be shorter. Then the elementary transformation is $y\mapsto xy$ (that is, replace $y$ with $xy'$ and since the variable $y$ is no more used, rename $y'$ back to $y$). If $x$ and $y$ have the same length, then $y$ is just renamed to $x$. Any solution can be obtained by a finite number of elementary transformations (see  \cite[Chapter 9.5]{lothaire}).

Defining properties of an anticongruence are that equivalent words have the same length and that factorizing an equivalence class yields equivalence classes again. It is important to observe that the sequence of transformations in the above algorithm, and hence the resulting solution, is fully determined by lengths of images, independently of whether we see them as words, or as classes. This is the basic reason why there is a tight correspondence between pseudo-solutions and ordinary solutions. 
\bigskip

I am grateful to an anonymous referees of earlier drafts of this paper who suggested several important improvements of the exposition. 

\bibliographystyle{plain}
\bibliography{pseudo}

\begin{thebibliography}{10}

\bibitem{lerest}
Evelyne Barbin-Le~Rest and Michel Le~Rest.
\newblock Sur la combinatoire des codes \`{a} deux mots.
\newblock {\em Theor. Comput. Sci.}, 41:61--80, 1985.

\bibitem{Czeizler:2011:ELR:1950991.1951213}
Elena Czeizler, Eugen Czeizler, Lila Kari, and Shinnosuke Seki.
\newblock {A}n extension of the {L}yndon--{S}ch\"{u}tzenberger result to
  pseudoperiodic words.
\newblock {\em Inf. Comput.}, 209(4):717--730, April 2011.

\bibitem{Elena2010}
Elena Czeizler, Lila Kari, and Shinnosuke Seki.
\newblock On a special class of primitive words.
\newblock {\em Theoretical Computer Science}, 411(3):617 -- 630, 2010.

\bibitem{DayWords2017}
Joel~D. Day, Pamela Fleischmann, Florin Manea, and Dirk Nowotka.
\newblock Equations {E}nforcing {R}epetitions {U}nder {P}ermutations.
\newblock In Sre{\v{c}}ko Brlek, Francesco Dolce, Christophe Reutenauer, and
  {\'E}lise Vandomme, editors, {\em Combinatorics on Words}, pages 72--84,
  Cham, 2017. Springer International Publishing.

\bibitem{dirkLS2005}
Tero Harju and Dirk Nowotka.
\newblock On the equation $x^k=z_1^{k_1}z_2^{k_2}\cdots z_n^{k_n}$ in a free
  semigroup.
\newblock {\em Theoretical Computer Science}, 330(1):117 -- 121, 2005.

\bibitem{CMUC}
{\v S}t{\v e}p{\'a}n Holub and Ji{\v r}{\'\i} S{\'y}kora.
\newblock Binary equality words with two $b$'s.
\newblock {\em Comment. Math. Univ. Carolin.}, 59(2):153 -- 172, 2018.

\bibitem{lothaire}
M.~Lothaire.
\newblock {\em Combinatorics on {W}ords}.
\newblock Cambridge Mathematical Library. Cambridge University Press, 2.
  edition, 1997.

\bibitem{ls}
R.~C. Lyndon and M.-P. Sch{\"u}tzenberger.
\newblock The equation $a^m=b^nc^p$ in a free group.
\newblock {\em The Michigan Mathematical Journal}, 9(4):289--298, 12 1962.

\bibitem{ManeaFW2012}
Florin Manea, Robert Merca{\c{s}}, and Dirk Nowotka.
\newblock Fine and {W}ilf's {T}heorem and {P}seudo-repetitions.
\newblock In Branislav Rovan, Vladimiro Sassone, and Peter Widmayer, editors,
  {\em Mathematical Foundations of Computer Science 2012}, pages 668--680,
  Berlin, Heidelberg, 2012. Springer Berlin Heidelberg.

\bibitem{PseudoLSComplete}
Florin Manea, Mike M{\"u}ller, Dirk Nowotka, and Shinnosuke Seki.
\newblock The extended equation of lyndon and schützenberger.
\newblock {\em Journal of Computer and System Sciences}, 85:132 -- 167, 2017.

\bibitem{handbook}
Grzegorz Rozenberg and Arto Salomaa, editors.
\newblock {\em Handbook of formal languages, vol. 1: word, language, grammar}.
\newblock Springer-Verlag New York, Inc., USA, 1997.

\bibitem{spehner}
J.-P. Spehner.
\newblock {\em {Quelques probl\`{e}mes d'extension, de conjugaison et de
  presentation des sous-mono\"{i}des d'un mono\"{i}de libre.}}
\newblock PhD thesis, Universit\'{e} Paris VII, Paris, 1976.

\end{thebibliography}

\end{document}